\title{Algorithmic Complexity of Isolate Secure Domination in Graphs}
\author{
        J. Pavan Kumar$ ^{1,*} $ \\[1pt]
            \and        P. Venkata Subba Reddy$ ^2 $\\[1pt]
        Department of Computer Science and Engineering\\
        National Institute of Technology\\
        Warangal, Telangana, India 506004. 
}
\date{}

\documentclass[12pt]{article}
\usepackage[a4paper,margin={3.5cm}]{geometry}
\usepackage{amsmath,amsthm,graphicx,amssymb}
\usepackage{float}
\usepackage{amsmath}
\usepackage{tikz}
\usetikzlibrary{plotmarks}
\usepackage{graphicx,subfigure}
\usepackage{algorithm,algpseudocode}

\makeatletter
\def\BState{\State\hskip-\ALG@thistlm}
\makeatother

\newtheorem{prop}{Proposition}

\newtheorem{Observation}{Observation}
\newtheorem{thm}{Theorem}
\newtheorem{cor}{Corollary}
\newtheorem{remark}{Remark}

\begin{document}
\maketitle

\begin{abstract}
A dominating set $S$ is an \emph{Isolate Dominating Set} (\textit{IDS}) if the induced subgraph $G[S]$ has at least one isolated vertex. In this paper, we initiate the study of new domination parameter called, \textit{isolate secure domination}. An isolate dominating set $S\subseteq V$ is an \textit{isolate secure dominating set} (\textit{ISDS}), if for each vertex $u \in V \setminus S$, there exists a neighboring vertex $v$ of $u$ in $S$ such that $(S \setminus \{v\}) \cup \{u\}$ is an IDS of $G$. The minimum cardinality of an ISDS of $G$ is called as an \emph{isolate secure domination number}, and is denoted by $\gamma_{0s}(G)$. Given a graph $ G=(V,E)$ and a positive integer $ k,$ the ISDM problem is to check whether $ G $ has an isolate secure dominating set of size at most $ k.$ We prove that ISDM is NP-complete even when restricted to bipartite graphs and split graphs. We also show that ISDM can be solved in linear time for graphs of bounded tree-width. 
\end{abstract}
\noindent\makebox[\linewidth]{\rule{0.65\paperwidth}{0.4pt}}
\noindent\makebox[\linewidth]{\rule{0.65\paperwidth}{0.4pt}}

\noindent * - Corresponding author\\
1 - jp.nitw@gmail.com\\
2 - pvsr@nitw.ac.in
\section{Introduction}
In this paper, every graph $G=(V,E)$ considered is finite, simple (i.e., without self-loops and multiple edges) and undirected with vertex set $V$ and edge set $E$.  
For a vertex $v \in V$, the (\textit{open}) \textit{neighborhood} of $v$ in $G$ is $N(v)= \{u \in V:(u,v) \in E$\}, the \textit{closed neighborhood} of $v$ is defined as $N[v]=N(v) \cup \{v\}$. If $S \subseteq V$, then the (open) neighborhood of $S$ is the set $N(S) = \cup_{v \in S} N(v)$. The closed neighborhood of $S$ is $N[S] = S \cup N(S)$. The \textit{degree} of a vertex $v$ is the size of the set $N(v)$ and is denoted by $d(v)$. If $d(v) = 0$, then $v$ is called an \textit{isolated vertex} of $G$. If $d(v) = 1$, then $v$ is called a \textit{pendant vertex}. For a graph $G=(V,E),$ and a set $S \subseteq V,$ the subgraph of $G$ \textit{induced} by $S$ is defined as $G[S]=(S,E_S)$, where $E_S=\{(x,y)$ $:$ $x,y \in S$ and $(x,y)\in E \}$. A \textit{spanning subgraph} is a subgraph that contains all the vertices of the graph. If $G[S]$ is a complete subgraph of $G$, then it is called a \textit{clique} of $G$. A set $S \subseteq V$ is an \textit{independent set} if $G[S]$ has no edges. A \textit{split graph} is a graph in which the vertices can be partitioned into a clique and an independent set. \par
In a graph $ G=(V,E) $, a set $S \subseteq V$ is a \textit{Dominating Set} (\textit{DS}) in $G$ if for every $u \in V \setminus S$, there exists $v \in S$ such that $(u,v) \in E$, i.e., $N[S] = V$. The minimum size of a dominating set in $G$ is called the \textit{domination number} of $G$ and is denoted by $\gamma(G)$. Given a graph $ G=(V,E)$ and a positive integer $ k,$ the domination decision (DOM) problem is to check whether $ G $ has a dominating set of size at most $ k.$ The DOM problem is known to be NP-complete \cite{garey}. The literature on various domination parameters in graphs has been surveyed in \cite{Haynes1,Haynes2}. An important domination parameter called secure domination has been introduced by E.J. Cockayne in \cite{pog}. A dominating set $S \subseteq V$ is a \textit{Secure Dominating Set} (\textit{SDS}) of $G$, if for each vertex $u \in V \setminus S$, there exists a neighboring vertex $v$ of $u$ in $S$ such that $(X \setminus \{v\}) \cup \{u\}$ is a dominating set of $G$ (in which case $v$ is said to \textit{defend} $u$). The minimum size of a SDS in $G$ is called the \textit{secure domination number} of $G$ and is denoted by $\gamma_s(G)$. Let $ S \subseteq V$. Then a vertex $ w \in V $ is called a private neighbor of $ v $ with respect to $ S $ if $ N[w] \cap S = \{v\}. $ If further $ w \in V \setminus S$, then $ w $ is called an \textit{external private neighbor (epn)} of $v$. The secure domination decision problem (SDOM) is known to be NP-complete for general graphs \cite{dev} and remains NP-complete even when restricted to bipartite graphs and split graphs \cite{osd}.\par Another domination parameter called isolate domination has been introduced by Hamid and Balamurugan in \cite{hamid}. A dominating set $S$ is an \emph{Isolate Dominating Set} (\textit{IDS}) if the induced subgraph $G[S]$ has at least one isolated vertex. The \textit{isolate domination number} $\gamma_0(G)$ is the minimum size of an IDS of $G$.  In \cite{rad}, N.J. Rad has proved that the isolate domination decision problem is NP-complete, even when restricted to bipartite graphs.  A dominating set $S$ is said to be a \textit{connected dominating set} (CDS), if the induced subgraph $G[S]$ is connected. A CDS $S$ is said to be a \textit{secure connected dominating set} (SCDS) in $G$ if for each $u \in V \setminus S$, there exists $v \in S$ such that  $uv \in E$ and $(S \setminus \{ v \})  \cup \{ u \} $ is a CDS in $G$. Algorithmic complexity of secure connected domination problem has been studied in \cite{jp1}. In this paper, we initiate the study of a variant of domination called \textit{isolate secure domination}. An isolate dominating set $S\subseteq V$ is an \textit{Isolate Secure Dominating Set} (\textit{ISDS}), if for each vertex $u \in V \setminus S$, there exists a neighboring vertex $v$ of $u$ in $S$ such that $(S \setminus \{v\}) \cup \{u\}$ is an IDS of $G$. The minimum size of an ISDS of $G$ is called as an \emph{isolate secure domination number}, and is denoted by $\gamma_{0s}(G)$. Given a graph $ G=(V,E)$ and a positive integer $ k,$ the ISDM problem is to check whether $ G $ has an isolate secure dominating set of size at most $ k.$ 
\paragraph{Motivation}A communication network is modeled as a graph $ G=(V, E) $ where each node represents a communicating device and each edge represents a communication link between two devices. All the nodes in the network need to communicate and exchange information to perform a task. However, in the networks where the reliability of the nodes is not guaranteed,  every node $ v $ can be a potential malicious node. A virtual protection device at node $ v $ can (i) detect the malicious activity at any node $ u $ in its closed neighborhood (ii) migrate to the malicious node $ u $ and repair it. One is interested to deploy minimum number of virtual protection devices such that every node should have at least one virtual protection device within one hop distance even after virtual protection  device is migrated to malicious node. This problem can be solved by finding a minimum secure dominating set of the graph $ G $. Further, if two virtual protection devices are deployed adjacently then there is a chance of one virtual protection device getting damaged or corrupted if other one is corrupted. 
In some scenarios it is desirable to have at least one virtual protection device isolated from all other devices before and after each potential migration to act as a backup device. This problem can be solved by finding a minimum isolate secure dominating set of the graph $ G $.

\section{Basic Results}
In this section, some precise values and bounds for new domination parameter called isolate secure domination in frequently encountered graph classes are presented. 
The graphs for which isolate secure dominating set exist the following observation holds.
\begin{Observation}\label{2}
	For a graph $ G $, $ \gamma_{0s}(G) \ge \gamma_{s}(G)$.
\end{Observation}
\begin{prop}\label{1}
	For the complete graph $ K_n $ with $ n $ vertices, $ \gamma_{0s}(K_n)=1.$
\end{prop}
\begin{thm}
	Let $ G $ be a graph with $ n $ vertices. Then $ \gamma_{0s}(G) =1$ if and only if $ G = K_n $.
\end{thm}
\begin{proof}
	Suppose $ \gamma_{os}(G) = 1 $ and let $ S = \{v\} $ be a minimum size ISDS of $ G $. Suppose $G \ne K_n$, then there exists $x, y \in V (G)$ such that $d(x, y) \ge 2$. Then $(S \setminus \{v\}) \cup \{x\} = \{x\}$, which is not a dominating set of $ G $, since $ (x, y) \notin E $. Therefore,
	$ G = K_n $. The converse is true by Proposition \ref{1}.
\end{proof}
\begin{remark} There is no isolate secure dominating set possible for the complete bipartite graph.
\end{remark}
\begin{prop}\label{3}
	Let $ P_n $ be a path graph with $ n$ $(\ge 4)$ vertices. Then $ \gamma_{0s}(P_n)=\lceil \frac{3n}{7} \rceil$.
\end{prop}
\begin{proof}
	From \cite{pog}, we know that $ \gamma_s(P_n)= \lceil \frac{3n}{7} \rceil.$ By this and Proposition \ref{2}, we have $ \gamma_{0s}(P_n)$ is at least $\lceil \frac{3n}{7} \rceil.$ Hence in order to complete the proof, we need to exhibit an ISDS of $ P_n $ of size $\lceil \frac{3n}{7} \rceil$. 
	\begin{figure}[H]
		\begin{center}
			\begin{tikzpicture}
			\draw (0,0)--(6,0);
			\draw (1,0) circle (1.5mm);
			\draw (3,0) circle (1.5mm);
			\draw (5,0) circle (1.5mm);
			
			\node at (0,0){\textbullet};	\node at (-0.1,-0.2){a};
			\node at (1,0){\textbullet};	\node at (0.9,-0.3){b};
			\node at (2,0){\textbullet};	\node at (1.9,-0.2){c};
			\node at (3,0){\textbullet};	\node at (2.9,-0.3){d};
			\node at (4,0){\textbullet};	\node at (3.9,-0.2){e};
			\node at (5,0){\textbullet};	\node at (4.9,-0.3){f};
			\node at (6,0){\textbullet};	\node at (5.9,-0.2){g};
			
			\end{tikzpicture}
		\end{center}
		\caption{ Path graph $ P_{7} $}
	\end{figure}
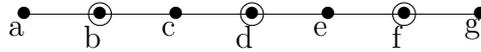
	Suppose $ P_n=(v_1,v_2,\ldots,v_n) $ and $ n=7m+r, $ where $ m\ge 1 $ and $ 0 \le r \le 6.$ Define two sets \[
	X = \bigcup_{i=0}^{m-1}\{v_{7i+2},v_{7i+4},v_{7i+6}\}
	\] and \[Y=
	\begin{cases}
	\emptyset, & 
	\text{if } r=0 \\
	\{v_{7m+1}\}, & \text{if } r=1 \text{ or } 2 \\
	\{v_{7m+1},v_{7m+3}\}, & \text{if } r=3 \text{ or }4 \\
	\{v_{7m+1},v_{7m+3},v_{7m+5}\}, & \text{if } r=5 \text{ or } 6 \\  
	\end{cases}
	\]
	Clearly $ \vert X \cup Y \vert = \lceil \frac{3n}{7} \rceil $. It can be noted that the set $ X \cup Y $ forms an independent set. It can also be verified that $ X \cup Y $ forms an ISDS of $ P_n $. Hence the result.   
\end{proof}
\begin{Observation}\label{pathcycle}
	If $ G^\prime $ is a spanning subgraph of graph $ G $, then $ \gamma_{0s}(G^\prime) \ge \gamma_{0s}(G) $.
\end{Observation}
\begin{prop}\label{cycle}
	Let $ C_n $ be a cycle graph with $ n $ vertices. Then $ \gamma_{0s}(C_n)=\lceil \frac{3n}{7} \rceil$. 
\end{prop}
\begin{proof}
	From Observation \ref{pathcycle}, $ \gamma_{0s}(P_n) \ge \gamma_{0s}(C_n)$.  From Proposition \ref{3}, it can be noted that $ \gamma_{0s}(C_n) \le \lceil \frac{3n}{7} \rceil $. From Observation \ref{2}, we know that $ \gamma_{s}(G) \le \gamma_{0s}(G)$. It is known that $ \gamma_s(C_n)= \lceil \frac{3n}{7} \rceil$ \cite{pog}. Hence the result.  
\end{proof}
\section{Main Results}
\noindent In this section, the complexity of a new domination parameter called isolate secure domination is investigated for bipartite graphs, split graphs and bounded tree-width graphs. \\[4pt]
The decision version of isolate secure domination problem is defined as follows.\\[4pt] 
\noindent
\textit{ISOLATE SECURE DOMINATION DECISION problem} (\textit{ISDM}) \\ [4pt]
\textit{Instance:} A simple, undirected graph $G=(V, E)$ and a positive integer $p$.\\
\textit{Question:} Does there exist a ISDS of size at most $ p $ in $ G $ ?\\[4pt]
The following proposition has been proved by Cockayne et al \cite{pog}.
\begin{prop}(\cite{pog})\label{p2}
	$X$ is a SDS of $G$ if and only if for each $u \in V \setminus X$, there exists $v \in X$ such that $G[epn(v,X) \cup \{u,v\}]$ is complete. 
\end{prop}
\noindent The following corollary follows from Proposition \ref{p2}.
\begin{cor}\label{cor}
	For any bipartite graph $G$ with SDS $S$, $\vert epn(v,S) \vert \le 1$, $\forall v \in V$.
\end{cor}
\noindent The DOM problem in bipartite graphs (DSDPB) has been proved as NP-complete by A.A. Bertossi \cite{bersto}.
\begin{thm}\label{isdmb}
	ISDM is NP-complete for bipartite graphs.
\end{thm}
\begin{proof} It can be shown that ISDM is in NP, since if a set $S \subseteq V$, such that $|S| \le k$ is given as a witness to a yes instance then it can be verified in polynomial time that $S$ is an ISDS of $G$. We reduce the DSDPB problem instance to ISDM problem instance for bipartite graphs as follows. Given a bipartite graph $G=(X, Y, E)$, we construct a graph $G^\prime=(X^\prime, Y^\prime, E^\prime)$ where $X^\prime= X$ $\cup$ $\{b_1,d_1,f_1,g_1,a_2,c_2,e_2\} $, $Y^\prime= Y \cup \{a_1,c_1,e_1,b_2,d_2,f_2,g_2\}$ and $E^\prime= E \cup \{ (a_1, v)$ $:$ $v \in X\} \cup \{ (a_2, v)$ $:$ $v \in Y\} \cup \{(a_i,b_i),(a_i,f_i),(a_i,g_i),(b_i,c_i)$, $(c_i,d_i),(d_i,e_i)$,
	$(e_i,f_i) : 1 \le i \le 2\}.$ Clearly, $G^\prime$ is a bipartite graph and can be constructed from $G$ in polynomial time. An example construction of a graph $G^\prime$ from a graph $G$ is depicted in figure \ref{fig:sub2}.\par
	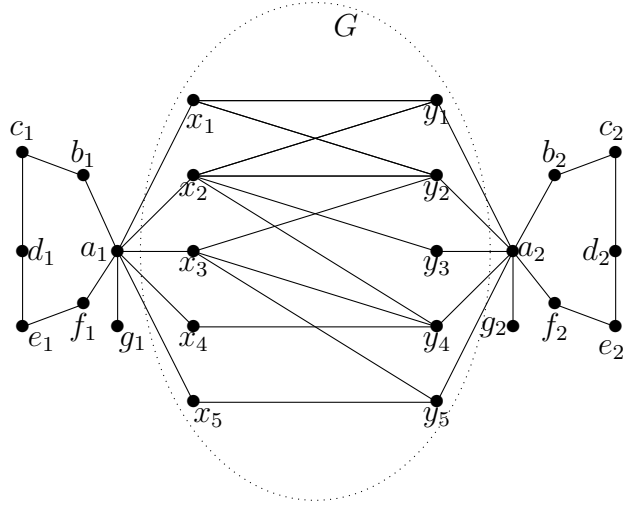
\begin{figure}
		\begin{center}
			\begin{tikzpicture}[scale=1.0]
			\node at (2,5){$ G $};
			\draw[dotted] (1.6,2) ellipse (2.3cm and 3.3cm);
			\node at (0,0){\textbullet}; \node at (0.2,-0.2){$ x_5 $};
			\node at (0,1){\textbullet}; \node at (0,0.8){$ x_4 $};
			\node at (0,2){\textbullet}; \node at (0,1.8){$ x_3 $};
			\node at (0,3){\textbullet}; \node at (0,2.8){$ x_2 $};
			\node at (0,4){\textbullet}; \node at (0.1,3.7){$ x_1 $};
			
			\node at (3.2,0){\textbullet}; \node at (3.2,-0.2){$ y_5 $};
			\node at (3.2,1){\textbullet}; \node at (3.2,0.8){$ y_4 $};
			\node at (3.2,2){\textbullet}; \node at (3.2,1.8){$ y_3 $};
			\node at (3.2,3){\textbullet}; \node at (3.2,2.8){$ y_2 $};
			\node at (3.2,4){\textbullet}; \node at (3.2,3.8){$ y_1 $};
			
			\draw (0,4)--(3.2,4)--(0,3)--(3.2,3)--(0,4);
			\draw (0,4)--(3.2,4)--(0,3)--(3.2,3)--(0,4);
			\draw (0,2)--(3.2,1)--(0,1);
			\draw (0,2)--(3.2,3);
			\draw (0,2)--(3.2,0);
			\draw (3.2,2)--(0,3)--(3.2,1);
			\draw (0,0)--(3.2,0);
			
			\node at (-1,2){\textbullet};\node at (-1.3,2){$ a_1 $};
			\node at (4.2,2){\textbullet};\node at (4.45,2){$ a_2 $};
			
			\draw (0,0)--(-1,2);
			\draw (0,1)--(-1,2);
			\draw (0,2)--(-1,2);
			\draw (0,3)--(-1,2);
			\draw (0,4)--(-1,2);
			
			\draw (3.2,0)--(4.2,2);												
			\draw (3.2,1)--(4.2,2);												
			\draw (3.2,2)--(4.2,2);												
			\draw (3.2,3)--(4.2,2);												
			\draw (3.2,4)--(4.2,2);												
			
			\node at (-1.45,3){\textbullet};\node at (-1.45,3.3){$ b_1 $};
			\node at (-2.25,3.3){\textbullet};\node at (-2.25,3.6){$ c_1 $};
			\node at (-2.25,2){\textbullet};\node at (-2.0,2){$ d_1 $};
			\node at (-2.25,1){\textbullet};\node at (-2.0,0.8){$ e_1 $};
			\node at (-1.45,1.3){\textbullet};\node at (-1.45,1){$ f_1 $};
			
			\node at (-1,1){\textbullet};\node at (-0.8,0.8){$ g_1 $};
			\draw (-1,2)--(-1,1);
			\node at (4.75,3){\textbullet};\node at (4.75,3.3){$ b_2 $};
			\node at (5.55,3.3){\textbullet};\node at (5.5,3.6){$ c_2 $};
			\node at (5.55,2){\textbullet};\node at (5.3,2){$ d_2 $};
			\node at (5.55,1){\textbullet};\node at (5.5,0.7){$ e_2 $};
			\node at (4.75,1.3){\textbullet};\node at (4.75,1){$ f_2 $};
			
			\node at (4.2,1){\textbullet};\node at (3.95,1){$ g_2 $};
			\draw (4.2,2)--(4.2,1);
			
			\draw (-1,2)--(-1.45,3)--(-2.25,3.3)--(-2.25,2)--(-2.25,1)--(-1.45,1.3)--(-1,2);
			\draw (4.2,2)--(4.75,3)--(5.55,3.3)--(5.55,2)--(5.55,1)--(4.75,1.3)--(4.2,2);
			%
			%
			\end{tikzpicture}
			\caption{Example construction of a graph $G^\prime$}\label{fig:sub2}
			
		\end{center}
	\end{figure}
	Next, we prove that $G$ has a dominating set of size at most $k$ if and only if $G^\prime$ has an ISDS of size at most $p=k+6$. Let $D$ be a dominating set of size at most $k$ in $G$ and $D^* = D \cup \lbrace a_1, c_1, e_1, a_2, c_2, e_2 \rbrace$. Clearly $D^*$ is an ISDS of size at most $k+6$ in $G^\prime$.\par
	Conversely, suppose that $G^\prime$ has an ISDS $D^*$ of size at most $p=k+6$. It can be observed that $\vert D^* \cap \{a_1,b_1,c_1,d_1,e_1,f_1\}\vert$ $\ge 3$, $\vert D^* \cap \{a_2,b_2,c_2,d_2,e_2,f_2\} \vert \ge 3$ and $\vert D^* \cap (X \cup Y) \vert \le k$. Since $g_1$ and $g_2$ are pendant vertices, $\vert D^* \cap \{a_1,g_1\} \vert \ge 1$ and $\vert D^* \cap \{a_2,g_2\} \vert \ge 1$. Let $D_1=D^* \cap \{a_1,g_1\}$ and $D_2=D^* \cap \{a_2,g_2\}$. Now we consider the following four possible cases. \\
	\textit{case }(\textit{i}). $\vert D_1 \vert =1$ and $\vert D_2 \vert =1$, then all vertices in $X \setminus D^*$ can be dominated by $Y \cap D^*$ and all vertices in $Y \setminus D^*$ can be dominated by $X \cap D^*$ hence $D^* \cap (X \cup Y)$ is a dominating set of size at most $k$ in $G$.\\
	\textit{case }(\textit{ii}). $\vert D_1 \vert =1$ and $\vert D_2 \vert =2$, then all vertices in $X \setminus D^*$ can be dominated by $Y \cap D^*$ and from Corollary \ref{cor}, $\vert epn(a_2,D^*) \vert \le 1,$ therefore $D^\prime \cup epn(a_2, D^*)$ is a dominating set of size at most $k$ in $G$.\\
	\textit{case }(\textit{iii}). $\vert D_1 \vert =2$ and $\vert D_2 \vert =1$, this case is analogous to case (ii).\\
	\textit{case }(\textit{iv}). $\vert D_1 \vert =2$ and $\vert D_2 \vert =2$, then if $D^* \cap (X \cup Y)$ is a dominating set of $G$, we are done. Otherwise, let $D^\prime=D^* \cap (X \cup Y)$. From Corollary \ref{cor}, $\vert epn(a_i, D^*) \vert \le 1,$ for $1\le i\le 2$ and hence $\vert D^\prime \vert \le k-2$. In such a case $D=D^\prime \cup \{epn(a_i, D^*) : 1 \le i \le 2 \}$ is a dominating set of size at most $k$ in $G$. Hence the proof. 
\end{proof}
\noindent The decision version of secure domination problem is defined as follows.\\[5pt]
\noindent
\textit{SECURE DOMINATION DECISION problem (SDOM)} \\ [6pt]
\textit{Instance:} A simple, undirected graph $G = (V, E)$ and a positive integer $k$.\\
\textit{Question:} Does there exist a SDS of size at most $ k $ in $ G $ ?\\[4pt]
It has been proved that the SDOM is NP-complete for split graphs by H.B. Merouane et al. \cite{osd}
\begin{thm}\label{isdms}
	ISDM is NP-complete for split graphs.
\end{thm}
\begin{proof} It is known that ISDM is in NP. We reduce SDOM problem instance for split graphs to ISDM problem instance for split graphs as follows. Given a split graph $G=(C, I, E)$, construct a graph $G^\prime= (C^\prime, I^\prime, E^\prime)$, where $C^\prime= C$ $\cup$ $\{x_1,y_1\} $, $I^\prime= I \cup \{x_2,y_2\}$ and $E(G^\prime)= E \cup \{ (x_1, v), (y_1, v) : v \in C\} \cup \{(x_1,y_1),(x_1,x_2),(y_1,y_2)\}.$ Note that $G^\prime$ is a split graph and can be constructed from $G$ in polynomial time. Figure \ref{fig:sub4} illustrates an example construction.\par
	\begin{figure}[H]					
		\begin{center}			
			\begin{tikzpicture}[scale=1.00]
			
			\node at (1,0){\textbullet}; \node at (0.7,0){$x_2$}; 
			\node at (4,0){\textbullet}; \node at (4.3,0){$y_2$}; 
			\node at (1,1.2){\textbullet}; \node at (0.7,1.2){$x_1$}; 
			\node at (4,1.2){\textbullet}; \node at (4.3,1.2){$y_1$};
			
			\node at (1,2.4){\textbullet}; \node at (0.85,2.4){$b$}; 
			\node at (4,2.4){\textbullet}; \node at (4.2,2.2){$c$};
			
			\node at (1,3.6){\textbullet}; \node at (0.7,3.6){$a$}; 
			\node at (4,3.6){\textbullet}; \node at (4,3.85){$d$};
			\node at (2.5,4.5){\textbullet}; \node at (2.2,4.5){$e$}; 
			
			\draw (1,0)--(1,1.2)--(1,2.4);
			\draw (4,0)--(4,1.2)--(4,2.4);
			
			\draw (1,3.6)--(1,2.4)--(4,2.4)--(4,3.6)--(2.5,4.5)--(1,3.6)--(4,3.6)--(1,2.4)--(2.5,4.5)--(4,2.4)--(1,3.6);
			\draw (1,1.2)..controls (0.45,2.4)..(1,3.6);
			\draw (1,3.6)--(4,1.2)--(1,1.2)--(4,2.4);
			\draw (1,1.2)--(4,3.6);
			\draw (1,1.2)--(2.5,4.5);
			\draw (4,1.2)..controls (4.45,2.4)..(4,3.6);
			\draw (4,1.2)--(2.5,4.5);						
			\draw (4,1.2)--(1,2.4);
			
			\node at (6.5,2.4){\textbullet}; \node at (6.75,2.4){$r$}; 
			\node at (6.5,3.6){\textbullet}; \node at (6.75,3.6){$q$}; 
			\node at (6.5,4.5){\textbullet}; \node at (6.75,4.5){$p$}; 
			
			\draw (6.5,2.4)--(1,3.6)--(6.5,4.5)--(4,3.6);
			\draw (4,2.4)--(6.5,3.6)--(2.5,4.5);
			
			\draw [dotted](3.8,3.6) ellipse (4cm and 2cm);
			\node at (4,5.3){$ G $};
			\end{tikzpicture}
			
			\caption{Example construction of a split graph $G^\prime$}\label{fig:sub4}	
		\end{center}
		
	\end{figure}
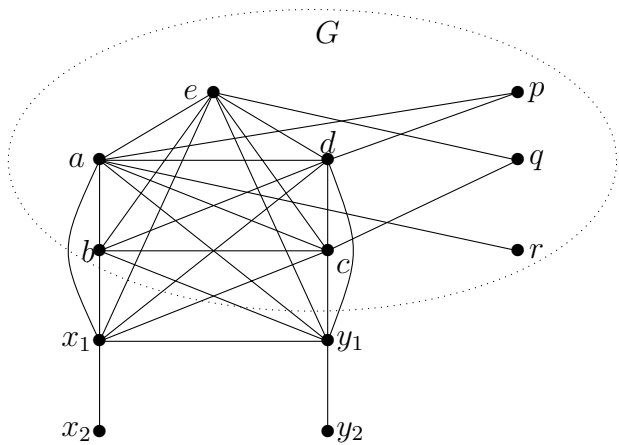
	
	Next, we shall show that $G$ has a SDS of size at most $k$ if and only if $G^\prime$ has an ISDS of size at most $p=k+2$. Let $D$ be a SDS of size at most $k$ in $G$ and $D^* = D \cup \lbrace x_2, y_2 \rbrace$. It can be easily verified that $D^*$ is an ISDS of size at most $k+2$ in $G^*$.\par
	Conversely, suppose that $G^\prime$ has an ISDS $D^*$ of size at most $p=k+2$. Let $X^\prime= D^* \cap \{x_1,x_2\}$, $Y^\prime= D^* \cap \{y_1,y_2\}$ and $D^\prime=D^* \setminus (X^\prime \cup Y^\prime)$. We now have the following cases. (i) $\vert X^\prime \vert = 1$ and  $\vert Y^\prime \vert = 1$, then $D^\prime$ is a SDS of size at most $k$ in $G$. (ii) $\vert X^\prime \vert = 1$ and  $\vert Y^\prime \vert = 2$, then for any $v \in C \setminus D^*$, $D^\prime \cup \{v\}$ is a SDS of size at most $k$ in $G$. (iii) $\vert X^\prime \vert = 2$ and  $\vert Y^\prime \vert = 1$, this case is similar to the previous case. (iv) $\vert X^\prime \vert = 2$ and  $\vert Y^\prime \vert = 2$, then for any two vertices $u,v \in C \setminus D^*$ and $u \ne v$, $D^\prime \cup \{u, v\}$ is a SDS of size at most $k$ in $G$.  
\end{proof}
\noindent Since the Domination problem is w[2]-complete for bipartite graphs and split graphs \cite{vraman} and the reductions in Theorems \ref{isdmb} and \ref{isdms} are in the function of the parameter $ k,$ the following two corollaries are immediate.
\begin{cor}
	ISDM is w[2]-hard for bipartite graphs.
\end{cor} 
\begin{cor}
	ISDM is w[2]-hard for split graphs.
\end{cor} 
\subsection{ISDM for bounded tree-width graphs}
Let $ G $ be a graph, $ T $ be a tree and $v\ $ be a family of vertex sets $ V_t \subseteq V (G) $ indexed by the vertices $ t $ of $ T $ . 
The pair $ (T, v\ ) $ is called a tree-decomposition of $ G $ if it satisfies the following three conditions:
(i) $ V(G)  =  \bigcup_{t \in V(T)} V_t $,
(ii) for every edge $ e \in E(G) $ there exists a $ t \in V(T)$ such that both ends of $ e $ lie in $ V_t $,
(iii) $V_{t_1} \cap   V_{t_3}  \subseteq V_{t_2}$ whenever $ t_1$, $t_2$, $t_3 \in V(T) $ and $ t_2 $ is on the path in $ T $ from $ t_1 $ to $ t_3 $.
The width of $ (T, v\ ) $ is the number $ max\{\vert V_t \vert-1 : t\in T \}$, and the tree-width $ tw(G) $ of $ G $ is the minimum width of any tree-decomposition of $ G $. By Courcelle's Thoerem, it is well known that every graph problem that can be described by counting monadic second-order logic (CMSOL) can be solved in linear-time in graphs of bounded tree-width, given a tree decomposition as input \cite{courc}. We show that ISDM problem can be expressed in CMSOL. 
\begin{thm}[\textit{Courcelle's Theorem}](\cite{courc})\label{cmsol1}
	Let $ P $ be a graph property expressible in CMSOL and let $ k $ be
	a constant. Then, for any graph $ G $ of tree-width at most $ k $, it can be checked in linear-time whether $ G $ has property $ P $.
\end{thm}
\begin{thm}\label{cmsol2}
	Given a graph $ G $ and a positive integer $ k $, ISDM can be expressed in CMSOL.
\end{thm}
\begin{proof}
	First, we present the CMSOL formula which expresses that the graph $ G $ has a dominating set of size at most $ k.$
	{\small 	$$Dominating(S)= (\vert S \vert \le k)  \land (\forall p)((\exists q)(q\in S \land adj(p,q))) \lor (p\in S)$$}
	where $ adj(p, q) $ is the binary adjacency relation which holds if and only if, $ p, q $ are two adjacent vertices of $ G.$
	$ Dominating(S) $ ensures that for every vertex $ p \in V $, either $ p\in S $ or $ p $ is adjacent to a vertex in $ S$ and the cardinality of $ S $ is at most $ k.$
	{\small 	$$Isolate-Dom(S)=Dominating(S) \land (\exists p)(p \in S \land ((\nexists q)(q\in (S \setminus \{p\}) \land adj(p,q))) $$} 
	Now, by using the above two CMSOL formulas we can express ISDM in CMSOL formula as follows.\\[4pt]
	{\small 	$ ISDM(S)=Isolate-Dom(S) \land$ \\ \hspace*{2.2cm}$(\forall x)((x \in S) \lor$ $((\exists y) (y \in S \land$ $Isolate-Dom((S \setminus \{y\}) \cup \{x\}) )))$}\\[4pt]
	Therefore, ISDM can be expressed in CMSOL.
\end{proof}
\noindent Now, the following result is immediate from Theorems \ref{cmsol1} and \ref{cmsol2}.
\begin{thm}
	ISDM can be solvable in linear time for bounded tree-width graphs. 
\end{thm}
\section{Conclusion}
\indent In this paper, it is shown that ISDM problem is NP-complete even when restricted to bipartite graphs, or split graphs. Since split graphs form a proper subclass of chordal graphs, this problem is also NP-complete for chordal graphs. We have shown that ISDM problem is linear time solvable for bounded tree-width graphs. It will be interesting to investigate the algorithmic complexity of ISDM problem for subclasses of chordal and bipartite  graphs. 

%
\end{document}